\documentclass[11pt]{article}
\usepackage{amssymb}
\usepackage{fullpage}
\usepackage{amsfonts}
\usepackage{amsmath}
\usepackage{cancel}
\usepackage{color}
\usepackage{fancyhdr}
\usepackage{ mathrsfs }
\usepackage{latexsym}
\usepackage{amsmath,amssymb,amsthm}
\usepackage{epsfig}
\usepackage{dsfont}

\newtheorem{definition}{Definition}

\newtheorem{thm}{Theorem}
\newtheorem{lemma}{Lemma}

\newtheorem{corollary}{Corollary}

\begin{document}

\title{Approximately Optimal Mechanisms for Strategyproof Facility Location: Minimizing $L_p$ Norm of Costs}

\author{
Itai Feigenbaum
\thanks{IEOR Department, Columbia University, New York,
NY;
{\tt itai@ieor.columbia.edu}} \and
Jay Sethuraman
\thanks{IEOR Department, Columbia University, New York,
NY;
{\tt jay@ieor.columbia.edu}. Research supported by NSF grant
CMMI-0916453 and CMMI-1201045.}
\and
Chun Ye
\thanks{IEOR Department, Columbia University, New York,
NY;
{\tt cy2214@columbia.edu}}
}

\date{September 2014}

\maketitle

\begin{abstract}
We consider the problem of locating a single facility on the real line.
This facility serves a set of agents, each of whom is located on the
line, and incurs a cost equal to his distance from the facility. An
agent's location is private information that is known only to him.
Agents report their location to a central planner who decides where to
locate the facility. The planner's objective is to minimize a "social"
cost function that depends on the agent-costs. However, agents might
not report truthfully; to address this issue, the planner must restrict
himself to {\em strategyproof} mechanisms, in which truthful reporting
is a dominant strategy for each agent. A mechanism that simply chooses
the optimal solution is generally not strategyproof, and so the planner
aspires to to use a mechanism that effectively {\em approximates} his
objective function. This general class of problems
was first studied by Procaccia and Tennenholtz and has been the subject
of much research since then.

In our paper, we study the problem described above with the social cost
function being the $L_p$ norm of the vector of agent-costs. We show
that the median mechanism (which is known to be strategyproof) provides
a  $2^{1-\frac{1}{p}}$ approximation ratio, and that is the optimal
approximation ratio among all deterministic strategyproof mechanisms.
For randomized mechanisms, we present two results. First, we  present a negative result: we show that for integer
$\infty>p>2$, no mechanism---from  a rather large class of randomized
mechanisms---
has an approximation ratio better than that of the median mechanism.
This is in contrast to the case of $p=2$ and $p=\infty$ where a randomized mechanism provably helps improve the worst case approximation ratio.
Second, for the case
of 2 agents, we show that a mechanism called LRM, first
designed by Procaccia and Tennenholtz for the special case of $L_{\infty}$,
provides the optimal  approximation ratio among all randomized
mechanisms. 
\end{abstract}

\section{Introduction}
We consider the problem of locating a single facility on the real
line. This facility serves a set of $n$ agents, each of whom is
located somewhere on the line as well. Each agent cares about his
distance to the facility, and incurs a disutility
(equivalently, cost) that is equal to
 his distance to access the facility. An agent's location is assumed
to be private information that is known only to him. Agents report
their locations to a central planner who decides where to locate
the facility based on the reports of the agents. The planner's objective
is to minimize a ``social'' cost function that depends on the vector of
distances that the agents need to travel to access the facility. It is
natural for the planner to consider locating the facility at a point
that minimizes her objective function, but in that case the agents may
not have an incentive to report their locations truthfully. As an example,
consider the case of 2 agents located at $x_1$ and $x_2$ respectively, and
suppose the location that optimizes the planner's objective
is the mid-point $ (x_1 + x_2)/2$. Then, assuming $x_1 < x_2$, agent 1
has an incentive to report a location $x'_1 < x_1$ so that the planner's
decision results in the facility being located closer to his true location.
The planner can address this issue by restricting herself to a {\em strategyproof}
mechanism: by this we mean that it should be a (weakly) dominant strategy
for each agent to report his location truthfully to the central planner.
This, of course, is an attractive property, but it comes at a cost: based on the
earlier example, it is clear that the planner cannot hope to optimize her
objective.
One way to avoid this difficulty is to assume an environment
in which agents (and the planner) can make or receive payments; in such
a case, the planner selects the location of the facility, and also a payment
scheme, which specifies the amount of money an agent pays (or receives) as a
function of the reported locations of the agents as well as the location of
the facility. This option gives the planner the ability to support the ``optimal''
solution as the outcome of a strategy-proof mechanism by constructing a carefully
designed payment scheme in which any potential benefit for a misreporting agent from 
a change in the location of the facility is offset by an increase in his payment.

There are many settings, however, in which such monetary compensations are either
not possible or are undesirable. This motivated 
Procaccia and Tennenholtz~\cite{pt13} to formulate the notion of {\em Approximate
Mechanism Design without Money}. In this model the planner restricts herself to
strategy-proof mechanisms, but is willing to settle for one that does not
necessarily optimize her objective. Instead, the planner's goal is to find a 
mechanism that effectively {\em approximates} her objective function. This is captured by
the standard notion of approximation that is widely used in the CS literature: for
a minimization problem, an algorithm is an $\alpha$-approximation if the solution it
finds is guaranteed to have cost at most $\alpha$ times that of the optimal 
cost ($\alpha \geq 1$).

Procaccia and Tennenholtz~\cite{pt13} apply the notion of approximate mechanism
design without money to the facility  location problem considered here
for two different objectives: (i) {\em minisum}, where the goal is to minimize the
sum of the costs of the agents; and (ii) {\em minimax}, where the goal is to minimize
the maximum agent cost. They show that for the minimax objective choosing any 
$k$-th median---picking the $k$th largest reported location---is a strategyproof, 
$2$-approximate mechanism. They design a
randomized mechanism called LRM (Left-Right-Middle) and show that it is a 
strategyproof, $3/2$-approximate mechanism; futhermore, they 
show that those mechanisms provide the optimal worst-case approximation 
ratio possible (among all deterministic and randomized strategyproof 
mechanisms, respectively). For the {\em minisum} objective, it is known 
that choosing the median reported location is optimal and 
strategyproof~\cite{moulin80}.
Feldman and Wilf~\cite{fw13} consider the same facility location problem
on a line but with the social cost function
being the $L_2$ norm of the agents' costs\footnote{Feldman and Wilf actually 
used the sum of squares of the agents' costs, 
but most of their results can be easily converted to the $L_2$ norm.
Of course, the  approximation ratios they report need to be adjusted as well.}.
They show that the median is 
a $\sqrt{2}$-approximate strategyproof mechanism for this objective 
function, and provide a randomized $(1+\sqrt{2}) / 2$-approximate 
strategyproof mechanism. In addition, 
facility location on other topologies such as circles and trees are considered by 
Alon et al. ~\cite{afpt09, afpt10b, afpt10a} as well as by Feldman and Wilf ~\cite{fw13}.

In our paper, we follow the suggestion of Feldman and Wilf~\cite{fw13}
and study the problem of locating a single facility on a line, but with the
objective function being the
$L_p$ norm of the vector of agent-costs (for general $p \geq 1$). We 
define the problem formally in section 2. In section 3, we show 
that the median mechanism (which is strategyproof) provides 
a $2^{1-\frac{1}{p}}$ approximation ratio, and that is the 
optimal approximation ratio among all deterministic strategyproof 
mechanisms. We move onto randomized mechanisms in section 4. First, we present a negative result: we show that for integer $\infty>p>2$, 
{\em no} mechanism---from 
a rather large class of randomized mechanisms---
has an approximation ratio better than that of the median mechanism, as the number of agents
goes to infinity.
It is worth noting that all the mechanisms proposed in literature so
far--- for minimax, minisum, as well as quadratic mean social 
cost functions--- belong to this class of mechanisms. Next, we consider the case of 2 agents,
and show that the LRM mechanism provides the optimal approximation 
ratio among all randomized mechanisms (that satisfy certain mild assumptions) for this special case, for every $p \geq 1$. 
Our result for the special case of 2 agents also gives a lower
bound on the approximation ratio for all randomized mechanisms. We briefly discuss some directions for
further research in section 5.

\section{Model}
Let $N = \{1, 2, \ldots, n\}$, $n \geq 2$, be the set of agents. Each agent $i \in N$
reports a location $x_i \in \mathbb{R}$. 
A {\em deterministic} mechanism is a collection of functions $f = \{f_n | \ n \in \mathbb{N},  n \geq 2\}$ such that each $f_n:\mathbb{R}^n \rightarrow \mathbb{R}$
maps each location profile ${\bf x} = (x_1, x_2, \ldots, x_n)$ to the 
location of a facility. We will abuse notation and let $f({\bf x})$ denote $f_n({\bf x})$.
Under a similar notational abuse, a {\em randomized} mechanism is a collection of functions $f$ that maps each location profile to
a probability distribution over $\mathbb{R}$: if $f(x_1, x_2, \ldots, x_n)$ is
the distribution $\pi$, then the facility is located by drawing a single sample
from $\pi$. 

Our focus will be on deterministic and randomized mechanisms for
the problem of locating a single facility when the location of any agent is
{\em private} information to that agent and cannot be observed or otherwise
verified. It is therefore critical that the mechanism 
be {\em strategyproof}---it should be optimal for each agent $i$ to report
his {\em true} location $x_i$ rather than something else. To that end we 
assume that if the facility is located at $y$, an agent's disutility, 
equivalently cost, is simply his distance to $y$. Thus, an agent whose true
location is $x_i$ incurs a cost $C(x_i,y) =  |x_i - y|$. If the location 
of the facility
is random and according to a distribution $\pi$, then the cost of agent $i$
is simply $C(x_i, \pi) = \mathbb{E}_{y \sim \pi} |x_i - y|$, 
where $y$ is a random variable
with distribution $\pi$. The formal definition of strategyproofness is now:

\begin{definition}
A mechanism $f$ is strategyproof if for each $i \in N$, 
each $x_i, x'_i \in \mathbb{R}$, and for each
${\bf x_{-i}} = (x_1, x_2, \ldots, x_{i-1}, x_{i+1}, \ldots x_n) \in \mathbb{R}^{n-1}$,
$$C(x_i, f(x_i, {\bf x_{-i}})) \leq C(x_i, f(x'_i, {\bf x_{-i}})), $$ 
where $(\alpha, {\bf x_{-i}})$ denotes a vector with the $i$-th component being $\alpha$ and the $j$-th component being $x_j$ for all $j \neq i$.
\end{definition}

The class of strategyproof mechanisms is quite large: for example, locating
the facility at agent 1's reported location is strategyproof, but is not
particularly appealing because it fails almost every reasonable notion of
fairness and could also be highly ``inefficient''. To address these issues,
and to winnow the class of acceptable mechanisms, we impose 
additional requirements that  stem from efficiency or fairness
considerations. In this paper we assume that locating a facility at $y$
for the location profile is ${\bf x} = (x_1, x_2, \ldots, x_n)$ incurs the
{\em social cost} 
$$sc({\bf x},y) \; = \; \bigg( \sum_{i \in N} |x_i - y|^{p} \bigg)^{1/p},\;\;\; p \geq 1.$$
For a randomized mechanism $f$ that maps $x$ to a distribution $\pi$,
we define the social cost to be
$$sc({\bf x},\pi) \; = \; \mathbb{E}_{y \sim \pi} \bigg( \sum_{i \in N} |x_i - y|^{p} \bigg)^{1/p}.$$
For this definition of social cost, our goal now is to find a strategyproof
mechanism that does well with respect to minimizing the social cost. A natural
mechanism (and this is the approach taken in the classical literature on
facility location) is the ``optimal'' mechanism: each location profile 
${\bf x} = (x_1, x_2, \ldots, x_n)$ is mapped to $OPT({\bf x})$, defined as\footnote{
Strictly speaking, the mechanism is not well defined in cases where
the social cost at $x$ is minimized by multiple locations $y$, but we could
pick an exogenous tie-braking rule to deal with such cases.}
$$OPT{(\bf x}) \; \in \; \arg \min_{y \in \mathbb{R}} sc({\bf x},y).$$
This optimal mechanism is not strategyproof as shown in the following
example.

\paragraph{Example.} Suppose there are two agents located at the points
$0$ and $1$ respectively on the real line. If they report their locations
truthfully, the optimal mechanism will locate the facility at $y = 0.5$,
for any $p > 1$.
Assuming agent 2 reports $x_2 = 1$, if agent 1 reports $x'_1 = -1$ instead, the facility will be located at 0, which
is best for agent 1.

\bigskip

Given that strategyproofness and optimality cannot be achieved simultaneously,
it is necessary to find a tradeoff. In this paper we shall restrict ourselves
to strategyproof mechanisms that approximate the optimal social cost as best
as possible. The notion of approxmation that we use is standard in computer
science: an $\alpha$-approximation algorithm is one that is guaranteed to have
cost no more than $\alpha$ times the optimal social cost. Formally, the
approximation ratio of an algorithm $A$ is $\sup_{I} \{ A(I) / OPT(I) \},$
where the supremum is taken over all possible instances $I$ of the problem;
and $A(I)$ and $OPT(I)$ are, respectively, the costs incurred by algorithm $A$
and the optimal algorithm on the instance $I$.

Our goal then is to design strategyproof (deterministic or randomized)
mechanisms whose approximation ratio is as close to 1 as possible.

\section{The Median Mechanism}

For the location profile ${\bf x} = (x_1, x_2, \ldots, x_n)$, the median mechanism
is a deterministic mechanism that locates the facility at the ``median''
of the reported locations. The median is unique if $n$ is odd, but not when 
$n$ is even, so we need to be more specific in describing the mechanism. 
For odd $n$, say $n = 2k-1$ for some $k \geq 1$, 
the facility
is located at $x_{[k]}$, where $x_{[k]}$ is the $k$th largest component of
the location profile. For even $n$, say $n = 2k$, the ``median'' can be any point in the
interval $[x_{[k]}, x_{[k+1]}]$; to ensure strategyproofness, we need to pick either
$x_{[k]}$ or $x_{[k+1]}$, and as a matter of convention we take the
median to be $x_{[k]}$.
It is well known that the median mechanism
is strategyproof~\footnote{A classical paper of Moulin~\cite{moulin80} for a closely
related model shows that all deterministic strategyproof mechanisms are 
essentially generalized median mechanisms.}. 
Furthermore, the median mechanism is {\em anonymous}\footnote{In an anonymous mechanism,
the facility location is the same for two location profiles that are permutations
of each other.}. Thus we may assume, without loss of generality, that each agent
reports her location truthfully.

Our main result in this section is that, for any $p \geq 1$, the median 
mechanism uniformly achieves 
the best possible approximation ratio among all deterministic 
strategyproof mechanisms.
We start with two simple observations, which will be used repeatedly in the
proof of this main result.

\begin{lemma}
\label{l:jen}
For any real numbers $a, b, c$ with $a \leq b \leq c$, and any $p \geq 1$,
$$(c-a)^p \; \leq \;  2^{p-1} [(c-b)^p + (b-a)^p].$$
\end{lemma}
\proof
For any $p \geq 1$, $f(x)=x^p$ is a convex function on $[0,\infty)$, and so for 
any $\lambda \in [0,1]$ and $x, y \geq 0$, 
\begin{equation}
f(\lambda x + (1 - \lambda) y )) \; \leq \; \lambda f(x) + (1-\lambda) f(y).
\end{equation}
Setting $\lambda = 1/2$, $x = c-b$, and $y = b-a$, and multiplying both
sides of the inequality by $2^p$ gives the result.
\qed

\begin{lemma}
\label{l:bin}
For any non-negative real numbers $a$ and $b$, and any $p \geq 1$,
$$(a+b)^p \; \geq \; a^p + b^p.$$
\end{lemma}
\proof
For integer $p$, the result is a direct consequence of the binomial theorem;
the same argument covers the case of rational $p$ as well. Continuity implies
the result for all $p$.
\qed

\begin{thm}
Suppose there are $n$ agents with the location 
profile ${\bf x} = (x_1, x_2, \ldots, x_n)$. Define
the social cost of locating a facility at $y$ as
$(\sum_{i=1}^n{|y-x_i|^p})^\frac{1}{p}$ for $p \geq 1$. 
The social cost incurred by the median mechanism is at most 
$2^{1-\frac{1}{p}}$ times the
optimal social cost \footnote{When $p = \infty$, the median mechanism provides a 2-approximation, as shown in  Procaccia and
Tennenholtz ~\cite{pt13}.}.
\end{thm}

\proof
We may assume that $x_1 \leq ... \leq x_n$. Let $OPT$ be a facility location that 
minimizes the social cost, and let $m$ be the median. The inequality we need to
prove is
$$\sum_{i=1}^n{|m-x_i|^p} \; \leq \; 2^{p-1} \sum_{i=1}^n{|OPT-x_i|^p}.$$
We do this by pairing each location $x_i$ with its ``symmetric'' location
$x_{n+1-i}$ and arguing that the total cost of these two locations in
the median mechanism is within the required bound of their total cost in
an optimal solution. For even $n$, this completes the argument; for odd $n$
the only location without such a pair is the median itself, which incurs {\em zero}
cost in the median mechanism, and so the argument is complete. 
Formally, the result follows if we can show
$$|m-x_i|^p+|x_{n+1-i}-m|^p \leq 2^{p-1}(|OPT-x_i|^p+|OPT-x_{n+1-i}|^p), \;\;\;
{\rm for all} \;\; i \leq \lfloor n/2 \rfloor.$$

We consider two cases, depending on whether $OPT$ is in the interval $[x_i, x_{n+1-i}]$
or not. In each of these cases, $OPT$ may be above the median or below, but the proof
remains identical in each subcase, so we give only one.

\begin{enumerate}
\item[1.] $x_i \leq m \leq OPT \leq x_{n+1-i}$ or $x_i \leq OPT \leq m \leq x_{n+1-i}$. 
We will prove the first of these subcases; the proof of the second is identical. 
Applying Lemma~\ref{l:jen} by setting $a = m$, $b = OPT$, and $c = x_{n+1-i}$, we get
$$ |x_{n+1-i}-m|^p \leq 2^{p-1}(|x_{n+1-i}-OPT|^p+|OPT-m|^p).$$
Thus,
\[%
\begin{split}
|m-x_i|^p+|x_{n+1-i}-m|^p & \leq |m-x_i|^p + 2^{p-1}(|x_{n+1-i}-OPT|^p+|OPT-m|^p) \\
& \leq 2^{p-1}(|m-x_i|^p + |x_{n+1-i}-OPT|^p+|OPT-m|^p) \\
& \leq 2^{p-1}(|x_{n+1-i}-OPT|^p+|OPT-x_i|^p),
\end{split}
\]
where the last inequality is obtained by applying Lemma~\ref{l:bin} to the terms $|m-x_i|^p$ and $|OPT-m|^p$.

\item[2.] $OPT \leq x_i \leq m \leq x_{n+1-i}$ or $x_i \leq m \leq x_{n+1-i} \leq OPT$. 
Again, we prove only the first subcase. Note that
\[%
\begin{split}
|x_{n+1-i}-m|^p + |m-x_{i}|^p &\leq |x_{n+1-i} - x_{i}|^p \\
& \leq  |OPT -x_{n+1-i}|^p \\
& \leq 2^{p-1}(|OPT - x_{i}|^p + |OPT - x_{n+1-i}|^p)
\end{split}
\]
where the first inequality follows from Lemma~\ref{l:bin}. (Note that Lemma~\ref{l:jen}
is not used in the proof of this case.)
\end{enumerate}

We end this section by showing that {\em no} deterministic and strategyproof 
mechanism can give a better approximation to the social cost.

\begin{lemma}
\label{l:2agents}
Consider the case of two agents and suppose the location profile is $(x_1, x_2)$
with $x_1 < x_2$. 
For $p \geq 1$, suppose
the social cost of locating a facility at $y$ is $(|x_1 - y|^p + |x_2 - y|^p)^{1/p}$.
Any determinstic mechanism whose approximation ratio is better than $2^{1-\frac{1}{p}}$ 
for $p > 1$  must locate the facility at $y$ for some $y \in (x_1, x_2)$.
\end{lemma}

\proof
The function $f(y) = |y-x_1|^p + |y-x_2|^p$ is strictly convex, and its 
unique minimizer is $y^* = (x_1 + x_2) / 2$, with the corresponding value 
$f(y^*) = |x_2 - x_1|^p / 2^{p-1}$.
Moreover $f(x_1) = f(x_2) = |x_2-x_1|^p = 2^{p-1}f(y^*)$. 
It follows that for the deterministic mechanism to do strictly better than
the stated ratio, the facility cannot be located at the reported locations;
locating the facility to the left of $x_1$ or to the right of $x_2$ only
increases the cost of the mechanism, so the only option left for a mechanism
to do better is to locate the facility in the interior, i.e., in $(x_1, x_2)$.
\qed

\begin{thm}
Any strategyproof deterministic mechanism has an approximation ratio of at 
least $2^{1-\frac{1}{p}}$ for the $L_p$ social cost function for any $p \geq 1$\footnote{The lower bound of $2$ on the approximation ratio also holds when $p = \infty$, see Procaccia and
Tennenholtz ~\cite{pt13}.}.
\end{thm}

\proof
The bound holds trivially for $p = 1$. Suppose $p > 1$, and suppose a deterministic
strategyproof mechanism yields an approximation ratio strictly better than 
$2^{1-\frac{1}{p}}$ to the $L_p$ social cost. For the two-agent location profile
$x_1 = 0, x_2 = 1$, 
Lemma~\ref{l:2agents} implies the facility is located at some $y \in (0,1)$.
Now consider the location profile $x_1 = 0, x_2 = y$. Again, by Lemma~\ref{l:2agents},
the mechanism locate the facility at $y' \in (0,y)$ to guarantee the improved
approximation. But if agent 2 is located at $y < 1$, he can misreport his location as
$1$, forcing the mechanism to locate the facility at $y$, his true location; this
violates strategyproofness.
\qed

\section{Randomized Mechanisms}

Recall that when the social cost is measured by the $L_{2}$ norm or the $L_{\infty}$
norm, randomization provably improves the approximation ratio. In the former case,
Feldman and Wilf~\cite{fw13} describe an algorithm whose approximation ratio is 
$(\sqrt{2}+1)/2$; for the latter, Procaccia and Tennenholtz~\cite{pt13} design an algorithm 
with an approximation ratio of $3/2$. The mechanisms in both cases are simple and somewhat
similar, placing non-negative probabilities \emph{only} on the optimal and reported locations, where these probabilities are independent of the reported location profile. There are two reasonable ways of choosing the reported locations: one is via dictatorships and the other is via generalized medians. In this section we show that neither is enough; namely, randomizing over dictatorships, generalized medians and the optimal location  
does not improve the approximation ratio of the median mechanism for {\em any} integer $p \in (2, \infty)$. 
For the case of 2 agents we show that the best approximation ratio is given by
the LRM mechanism among all strategyproof mechanisms. Extending this analysis
even to the case of 3 agents appears to be non-trivial.

\subsection{Mixing Dictatorships and Generalized Medians with the Optimal Location}

\begin{thm}
Suppose we are given non-negative numbers ${p_j^n}'$, ${p_j^n}''$, and $p_{OPT}^n$ with
$p_{OPT}^n+\sum_{j \in N}{p_j^n}'+\sum_{j \in N}{p_j^n}''=1$ for each $n$.
For the problem with $n$ agents and reported profile $(x_1, x_2, \ldots, x_n)$
let $f$ be the strategyproof randomized mechanism that locates the facility at $OPT$ with probability $p_{OPT}^n$, at $x_j$ with probability ${p_j^n}'$, and at $x_{[j]}$ with probability ${p_j^n}''$\footnote{When a location appears more than once in $x_1,\ldots,x_n,x_{[1]},\ldots,x_{[n]},OPT$, the probabilities add up.}, where $OPT$ is the optimal location for the profile $(x_1, x_2, \ldots, x_n)$.
Then, for any finite integer $p>2$, the approximation ratio of $f$ is at least $2^{1-\frac{1}{p}}$.
\end{thm}

\begin{proof}
Fix $n=2k$, with $k \in \mathbb{N}$. For notational convenience, let $p_j^n={p_j^n}'+{p_j^n}''$ (we will use this to analyze profiles where $x_j=x_{[j]}$). For $j=1,\ldots,k$, define the profile ${\bf x}^j$ as follows (where $a_j$ is a parameter to be defined shortly): agents $1$ through $j$ are located at $-a_j$; agents $j+1$ through $k$ are located at $0$; agents $k+1$ through $2k-j+1$ are located at $1$; and agents $2k-j+2$ through $2k$ are located at $1+a_j$. (Note the {\em slight} asymmetry in the location of the agents: while $k$ agents are at or below zero, and $k$ agents are at or above $1$, there is an additional agent at 1 compared to zero and so one less agent at $1+a_j$ compared to $-a_j$.)
Now, $a_j$ is chosen to be the smallest positive root of the function $g_j(\alpha)=j\alpha^{p-1}-(k-j+1)-(j-1)(1+\alpha)^{p-1}$; such an $a_j$ must exist by the intermediate value theorem, as $g_j(0) < 0$ and $g_j(\alpha)$  is a continuous function of $\alpha$ with $g_j(\alpha) \rightarrow \infty$ as $\alpha \rightarrow \infty$.

We show that the optimal mechanism locates the facility at zero for the profile ${\bf x}^j$, i.e., $OPT = 0$. Note that the social cost for this profile, when locating the facility at $z \in [0,1]$, is $j(z+a_j)^p+(k-j)z^p+(k-j+1)(1-z)^p+(j-1)(1+a_j-z)^p$, and when $z \in (-a_j,0)$ the social cost becomes $j(z+a_j)^p+(k-j)(-z)^p+(k-j+1)(1-z)^p+(j-1)(1+a_j-z)^p$. Note that the social cost function is differentiable for $z \in (0,1)$ and for $z\in (-a_j,0)$. The left and right derivatives at $0$ are both $pja_j^{p-1}-p(k-j+1)-p(j-1)(1+a_j)^{p-1}$, and thus the social cost function is differentiable on $(-a_j,1)$ with its derivative at $z=0$ equal to zero (by our choice of $a_j$). The social cost function can also easily be verified to be twice differentiable on $(-a_j,1)$, with a positive second derivative at $z=0$, and thus we have a local minimum at $z=0$. The fact that this is a global minimum now follows from strict convexity of the social cost function $||{\bf x}^j-z(1,\ldots,1)||_p$ (for all $z \in \mathbb{R}$). Thus, indeed, $OPT=0$.

We now attempt to bound $p_{OPT}$. For each profile ${\bf x}^j$, consider the profile ${{\bf x}^j}'$ that differs only in the location of agent $j$: namely, ${x_j^j}'=0$ instead of $-a_j$. Note that on this profile, $OPT=0.5$ by symmetry. Strategyproofness implies that a deviation from profile ${{\bf x}^j}'$ to profile ${\bf x}^j$ should not be beneficial for agent $j$, namely $a_jp_j^n-\frac{1}{2}p_{OPT}^n \geq 0$ (where $a_j$ is the increase in agent $j$'s cost caused by that deviation when the facility is built in his reported location, and $\frac{1}{2}$ is the decrease in his cost caused by that deviation when the facility is located at $OPT$), which implies $p_j^n \geq \frac{p_{OPT}^n}{2a_j} $. Defining $a_j$ for $j=k+1,\ldots,2k$ in a symmetric fashion, we see that the same inequality holds for $j$ in that range, and that $a_j=a_{2k-j+1}$. Summing those inequalities up, we get:

$$1-p_{OPT}^n=\sum_{j=1}^{2k} p_j^n \geq \sum_{j=1}^{2k} \frac{p_{OPT}^n}{2a_j} =2\sum_{j=1}^k \frac{p_{OPT}^n}{2a_j} =\sum_{j=1}^k\frac{p_{OPT}^n}{a_j} $$

$$p_{OPT}^n \leq \frac{1}{1+\sum_{j=1}^k \frac{1}{a_j}}$$

Now, we claim it is enough to show that as $n \rightarrow \infty$ (or equivalently, as $k \rightarrow \infty$), $\sum_{j=1}^k \frac{1}{a_j} \rightarrow \infty$. The inequality then implies that $p_{OPT}^n \rightarrow 0$. Consider the profile which locates $k$ agents at $0$ and $k$ agents at $1$. The social cost of locating the facility at $OPT$ on this profile is $\sqrt[p]{n}/2$, while the social cost of locating the facility at an agent's location is $\sqrt[p]{n}2^{-\frac{1}{p}}$; thus, the approximation ratio of $f$ on this profile is $\frac{p_{OPT}^n \sqrt[p]{n}/2+(1-p_{OPT}^n)\sqrt[p]{n}2^{-\frac{1}{p}}}{\sqrt[p]{n}/2}=2^{1-\frac{1}{p}}-(2^{1-\frac{1}{p}}-1)p_{OPT}^n$. Thus, as $n \rightarrow \infty$, the approximation ratio on these profiles approaches $2^{1-\frac{1}{p}}$, completing the proof.\\

We are left with the task of showing that $\lim_{k \rightarrow \infty} \sum_{j=1}^k \frac{1}{a_j}=\infty$ . To do so, we first show that for $j \geq k^\frac{1}{p-1}+1$, $2^{p-1}(j-1) > a_j$. Recall that $a_j$ was defined as the smallest positive root of $g_j(\alpha)$, and that $g_j(0)<0$. Thus, it is enough to show that for $j$ in the appropriate range, $g_j(2^{p-1}(j-1))>0$. For notational convenience, we denote $Q=2^{p-1}$.

\[
\begin{split}
g_j(Q(j-1)) &= jQ^{p-1}(j-1)^{p-1}-(k-j+1)-(j-1)(1+Q(j-1))^{p-1} \\
&= Q^{p-1}(j-1)^{p-1}-k-(j-1)\sum_{i=1}^{p-2}\binom{p-1}{i}(Q(j-1))^{p-1-i} \\
&\geq Q^{p-1}(j-1)^{p-1}-(j-1)^{p-1}-(j-1)\sum_{i=1}^{p-2}\binom{p-1}{i}(Q(j-1))^{p-1-i} \\
&> Q^{p-1}(j-1)^{p-1}-(j-1)^{p-1}-(j-1)\sum_{i=1}^{p-2}\binom{p-1}{i}(Q(j-1))^{p-2} \\
&> Q^{p-1}(j-1)^{p-1}-(j-1)\sum_{i=1}^{p-1}\binom{p-1}{i}(Q(j-1))^{p-2} \\
&= Q^{p-1}(j-1)^{p-1}-(j-1)(Q(j-1))^{p-2}\sum_{i=1}^{p-1}\binom{p-1}{i} \\
&> Q^{p-1}(j-1)^{p-1}-(j-1)(Q(j-1))^{p-2}2^{p-1}= 0. \\
\end{split}
\]

Now,

\[
\begin{split}
\lim_{k \rightarrow \infty} \sum_{j=1}^k \frac{1}{a_j}&>\lim_{k \rightarrow \infty} \sum_{j=\lceil{}k^\frac{1}{p-1}+1\rceil{}}^k \frac{1}{2^{p-1}j} \\
&=\frac{1}{2^{p-1}}\lim_{k \rightarrow \infty} \sum_{j=\lceil{}k^\frac{1}{p-1}+1\rceil{}}^k \frac{1}{j} \\
&\geq \frac{1}{2^{p-1}}\lim_{k \rightarrow \infty} \int_{k^\frac{1}{p-1}+2}^k \frac{1}{t} dt \\
&=\frac{1}{2^{p-1}}(\lim_{k \rightarrow \infty} \int_{k^\frac{1}{p-1}}^k \frac{1}{t} dt-\lim_{k \rightarrow \infty}\int_{k^\frac{1}{p-1}}^{k^\frac{1}{p-1}+2}\frac{1}{t}dt) \\
&=\frac{1}{2^{p-1}}((\lim_{k \rightarrow \infty}(1-\frac{1}{p-1})\ln{k})-0)=\infty \\
\end{split}
\]

which completes our proof.
\end{proof}

\subsection{Optimality of the LRM Mechanism for 2 Agents}
Our next result shows that the LRM mechanism provides the best possible approximation ratio among all shift and scale invariant (defined below) startegyproof mechanisms for the case of $2$ agents for all $L_p$ social cost functions for $p \geq 1$.

We begin with some definitions: we say that a mechanism $f$ is \emph{shift and scale invariant} if for every location profile ${\bf x}=\{x_1,x_2\}$ s.t. $x_1 \leq x_2$ and every $c \in \mathbb{R}$, the following two properties are satisfied:
\begin{enumerate}
\item $f(\{x_1+c,x_2+c\})=f({\bf x})+c$.
\item When $c \geq 0$, we have $f(\{cx_1,cx_2\})=cf({\bf x})$, and when $c < 0$, we have $f(\{cx_1,cx_2\})=-cf(\{-x_2, -x_1\})$.
\end{enumerate}
A convenient notation for a given location profile ${\bf x}$ is to denote its midpoint as $m_{\bf x}=\frac{x_1+x_2}{2}$. We say that a mechanism $f$ is \emph{symmetric} if for any location profile ${\bf x}$ and for any $y \in \mathbb{R}$, $\mathds{P}(f({\bf x}) \geq m_{\bf x} + y) = \mathds{P}(f({\bf x}) \leq m_{\bf x} - y)$.

The following lemma allows us to convert any strategyproof mechanism into a symmetric mechanism.

\begin{lemma}
Given any strategyproof mechanism, there exists another symmetric strategyproof mechanism with the same approximation ratio.
\end{lemma}

\begin{proof}
Given a mechanism  $f$, we define the \emph{mirror mechanism} of $f$, $f_{mirror}$, to be such that for any profile $\bf x$, we have that  $\mathds{P}(f_{mirror}({\bf x}) \geq m_{\bf x}+b)=\mathds{P}(f({\bf x}) \leq m_{\bf x}-b)$ for all $b \in \mathbb{R}$ and location profiles ${\bf x}$. \\ 
Assume $f$  is a strategyproof mechanism. Symmetry dictates that $f_{mirror}$ must also be strategyproof, since
any misreport of the right agent with respect to $f$ induces the same cost as that of an equivalent a misreport of the left agent with respect to $f_{mirror}$ and vice versa. Moreover, since composing two strategproof mechanisms yields a strategyproof mechanism, the mechanism $g$ that picks $f$ with probability $1/2$ and $f_{mirror}$ with probability 1/2 is a strategyproof mechanism that is also symmetric. Finally, note that $g$ has the same approximation ratio as $f$ for all location profiles, since $f_{mirror}$ has the same approximation ratio as $f$.
\end{proof}

From now on, whenever we talk about a shift and scale invariant mechanism, we will also assume that it is symmetric. To simplify our proof of the main result, we will assume in addition that given a reported profile ${\bf x} = \{x_1, x_2\}$, the mechanism will only assign a facility location that lies in between $x_1$ and $x_2$, i.e. $\mathbb{P}(y \in [x_1, x_2]) = 1$, where $y$ is a random variable representing the facility location assigned by the mechanism. It is worth noting that the main result remains true even without this assumption, although the complete proof is somewhat long and cumbersome, so we will omit it. The next lemma deals with an equivalent condition for strategyproofness with respect to a shift, scale invariant and symmetric mechanism.

\begin{lemma}
A shift, scale invariant, and symmetric mechanism $f$ is strategyproof if and only if for any profile ${\bf x} = \{x_1, x_2\}$ with $x_1 = 0 < x_2$, the following condition holds:

- \[ \int_{(-\infty, x_2)} y dF(y) + \int_{(x_2, \infty)} y dF(y) + x_2\mathds{P}(Y=x_2) \geq 0,\]

where $Y = f(\bf{x})$ with c.d.f. $F$.
\end{lemma}
\begin{proof}
By shift invariance, it suffices to check strategyproofness for profiles where $x_1=0$ and by symmetry, we can assume without lost of generality that $x_2 \geq 0$. Moreover, any shift, scale invariant mechanism is trivially strategyproof with respect to the profile $\{0, 0\}$ since by definition, the mechanism would place all of the probability mass on $0$, which means that no agent has incentive to misreport his location. Thus, we can assume that $x_2 > 0$.\\ Since the mechanism is symmetric, it suffices to show that agent 2 cannot benefit by deviating from his true location if and only if the aforementioned condition hold. Since $x_2 > 0$, we can denote agent 2's deviation $x'_2$ as $cx_2$ for some $c \in \mathbb{R}$. Moreover, since $\mathbb{P}(Y \in [x_1, x_2]) = 1$ by assumption, we can further restrict ourselves to the case where  $c > 1$ because agent 2 has no incentive to deviate to a location $cx_2$ where $cx_2 < x_2$ as the mechanism is scale invariant.

When agent 2 reports his location to be $cx_2$, where $c > 1$, the change in cost incurred by agent 2 is: 
\[%
\begin{split}
C_{dev} - C_{orig} & =-(c-1) \int_{(-\infty, \frac{x_2}{c})} y dF(y) + \int_{[\frac{x_2}{c}, x_2)}((c+1)y - 2x_2)dF(y) +  (c-1)\int_{(x_2, \infty)} ydF(y) \\
& + (c-1)x_2\mathds{P}(Y=x_2) \\
& = -(c-1) \int_{(-\infty, x_2)} y dF(y) + \int_{[\frac{x_2}{c}, x_2)}(2cy - 2x_2)dF(y) +  (c-1)\int_{(x_2, \infty)} ydF(y)  + \\ &(c-1)x_2\mathds{P}(Y=x_2) \\
& \geq -(c-1) \int_{(-\infty, x_2)} y dF(y) +  (c-1)\int_{(x_2, \infty)} ydF(y)  + (c-1)x_2\mathds{P}(Y=x_2)  \\
\end{split}
\]
Hence, when condition $1$ holds, we have that $ -(c-1) \int_{(-\infty, x_2)} y dF(y) +  (c-1)\int_{(x_2, \infty)} yF(y)  + (c-1)x_2\mathds{P}(Y=x_2) \geq 0$, which means that $C_{dev} - C_{orig} \geq 0$. \\

To prove the other direction, suppose the condition does not hold, then there exists $\epsilon > 0$ small enough such that $- \int_{(-\infty, x_2)} y dF(y) + \int_{(x_2, \infty)} y dF(y) + x_2\mathds{P}(Y=x_2) \leq - \epsilon$ for some $x_2 > 0$. We choose $c > 1$ s.t. $\mathds{P}(Y \in [\frac{x_2}{c}, x_2)) < \frac{\epsilon}{4x_2} $, then we have that

\[%
\begin{split}
C_{dev} - C_{orig} & = -(c-1) \int_{(-\infty, x_2)} y dF(y) + \int_{[\frac{x_2}{c}, x_2)}(2cy - 2x_2)dF(y) +  (c-1)\int_{(x_2, \infty)} ydF(y)  \\
&+ (c-1)x_2\mathds{P}(Y=x_2) \\
& \leq (c-1)(-\int_{(-\infty, x_2)} y dF(y) + \int_{[\frac{x_2}{c}, x_2)}(2x_2)dF(y) +  \int_{(x_2, \infty)} ydF(y)  + x_2\mathds{P}(Y=x_2) ) \\
& < -(c-1)\frac{\epsilon}{2} < 0
\end{split}
\]
which contradicts strategyproofness of the mechanism. \footnote{Notice that given any shift, scale invariant, and symmetric mechanism $f$, in order to check whether $f$ is a strategyproof mechanism, it suffices to check whether $f$ is strateyproof for one particular profile. Without lost of generality, we can assume that $x_1 = 0$ and $x_2 = 1$. Here is a short proof of the claim. By the same argument as before, it suffices to check strategyproofness for all profiles $\{x_1, x_2\}$, where $x_2 > x_1 = 0$. Let $Y = f(\{0,1\})$, then $f(\{0, x_2\}) = x_2Y$. The mechanism is strategyproof with respective to the profile $\{x_1, x_2\}$ if and only if for all $c \in \mathbb{R}$, we have that \[E[|cx_2Y - x_2|] \geq E[|x_2Y - x_2|].\] Since $x_2 >0$, this follows directly from the strategyproofness condition for the profile $\{0,1\}$: 
\[E[|cY - 1|] \geq E[|Y -1|]  \ \forall c \in \mathbb{R}.\]}
\end{proof}

Given a strategyproof, shift, scale invariant and symmetric mechanism, the next lemma demonstrates how to find another strategyproof, shift, scale invariant and symmetric mechanism that restricts the probability assignment to $x_1, x_2$, and $m_{\bf x}$ for all profile ${\bf x}$ and simultaneously gives a better approximation than the original mechanism.

\begin{lemma}
Let $f$ be a strategyproof shift, scale invariant  and symmetric mechanism, where $\mathds{P}(f({\bf x}) \in [x_1, x_2]) = 1$ for location profile ${\bf x} = \{x_1, x_2 \}$ with $x_2 > x_1$, then there exists another strategyproof mechanism $g$ such that $\mathds{P}(g({\bf x}) \in \{x_1, x_2, m_{\bf x}\}) = 1$ for the (and thus every) location profile ${\bf x}$ and that  $E[sc(g({\bf x}), {\bf x})] \leq E[sc(g({\bf x}), {\bf x})]$.  Furthermore, $g$ satisfies shift, scale invariance and symmetry.\end{lemma}
\begin{proof}
Now, let $g$ be the mechanism that satisfies $\mathds{P}(g({\bf x})=x_1)=\mathds{P}(f({\bf x})=x_1)$, $\mathds{P}(g({\bf x})=x_2)=\mathds{P}(f(x)=x_2)$, $\mathds{P}(g({\bf x})=m_x)=1-\mathds{P}(g({\bf x})=x_1)-\mathds{P}(g({\bf x})=x_2)$. Note that since $m_x$ minimizes the social cost function for the profile ${\bf x}$, $g$ certainly provides a weakly better approximation ratio than $f$. 
By shift invariance, we can assume wlog that $0 = x_1 \leq x_2$, then an alternative way to show strategyproofness is to check to see that the condition of the lemma 1 is satisfied by $g$.  Since $f$ is a strategyproof mechanism, the condition implies that 
\[%
\begin{split}
0 & \leq - \int_{(0, x_2)} y d(F(y)) + x_2\mathds{P}(f({\bf{x}})=x_2)  \\
& = -\int_{(-\frac{x_2-x_1}{2}, \frac{x_2-x_1}{2})} (m_{\bf{x}} + u) d(F(m_{\bf{x}} + u)) + x_2\mathds{P}(f({\bf{x}})=x_2) \\
& = - m_{\bf{x}}\mathds{P}({f({\bf x}) \in (x_1,x_2)}) + -\int_{(-\frac{x_2-x_1}{2}, \frac{x_2-x_1}{2})}ud(F(m_{\bf{x}} + u)) +  x_2\mathds{P}(f({\bf{x}})=x_2) \\
& = - m_{\bf{x}}(1-\mathds{P}(g({\bf x})=x_1)-\mathds{P}(g({\bf x})=x_2)) +  x_2\mathds{P}(g({\bf{x}})=x_2) 
\end{split}
\]
Note, $ -\int_{(-\frac{x_2-x_1}{2}, \frac{x_2-x_1}{2})}ud(F(m_{\bf{x}}+u)) = 0$ because the distribution is symmetric around $m_{\bf{x}}$. Hence, the condition is satisfied for the mechanism $g$.
\end{proof}

Thus, $g$ is a symmetric strategyproof mechanism that provides a weakly better approximation ratio than $f$ and which satisfies $\mathds{P}(g({\bf x}) \in \{x_1,x_2,m_{\bf x}\})=1$ for every location profile $\bf x$. \\

\noindent Now we are ready to prove the main theorem.
\begin{thm} \label{4}
The LRM mechanism gives the best approximation ratio among all strategyproof mechanisms that are shift and scale invariant. 
\end{thm}
\begin{proof}
By the previous lemma, it suffices to search among the class of strategyproof shift, scale invariant  and symmetric mechanisms where any element $f$ of the class satisfies the property that $\mathds{P}(f({\bf x}) \in \{x_1,x_2,m_{\bf x}\})=1$. It is not difficult to see that in order to enforce strategyproofness, we must have that $\mathds{P}(f({\bf x}) \in \{x_1,x_2 \}) \geq 0.5$, which implies that among all such mechanisms, LRM provides the best approximation ratio of $0.5(2^{1-\frac{1}{p}} + 1)$.
\end{proof}

An immediate consequence of Theorem~\ref{4} is the following corollary.

\begin{corollary}
Any strategyproof shift and scale invariant mechanism has an approximation of at least $0.5(2^{1-\frac{1}{p}} + 1)$ in the worst case.
\end{corollary}

\section{Discussion}
The most important open question in our view is whether or not randomization
can help improve the worst-case approximation ratio for general $L_p$ norm
cost functions. The case of $p = 1$ is uninteresting because there is an
optimal deterministic mechanism; for $p = 2$ and $p = \infty$ we already
saw that randomization improves the worst-case approximation ratio, but
we do not know if this is simply a happy coincidence, or if one can obtain
similar results for all $p > 2$. Our negative result in Section 4 implies that any improvement by randomization would require a different approach than the existing mechanisms.

There are many other natural questions as well: for instance, what happens
for more general topologies such as trees or cycles? Is it possible to
characterize all randomized strategy-proof mechanisms on specific topologies?

Finally, we believe it is of interest to consider more general cost functions for the individual agents. The properties established for
the LRM and many other randomized mechanisms depend on the assumption that agents
incur costs that are \emph {exactly equal} to the distance to access the
facility. Clearly, this is a very restrictive assumption, and working
with more general individual agent costs is a promising direction
to broaden the applicability of this class of models.\footnote{For deterministic mechanisms, our result continues to hold for arbitrary single peaked cost functions, as long as the social cost remains an $L_p$ measure of the distances.}


\bibliographystyle{plain}
\bibliography{newfl}

\begin{thebibliography}{1}

\bibitem{afpt09}
Noga Alon, Michal Feldman, Ariel~D. Procaccia, and Moshe Tennenholtz.
\newblock Strategyproof approximation mechanisms for location on networks.
\newblock {\em CoRR}, abs/0907.2049:3432--3435, 2009.

\bibitem{afpt10b}
Noga Alon, Michal Feldman, Ariel~D. Procaccia, and Moshe Tennenholtz.
\newblock Strategyproof approximation of the minimax on networks.
\newblock {\em Math. Oper. Res.}, 35(3):513--526, 2010.

\bibitem{afpt10a}
Noga Alon, Michal Feldman, Ariel~D. Procaccia, and Moshe Tennenholtz.
\newblock Walking in circles.
\newblock {\em Discrete Mathematics}, 310(23):3432--3435, 2010.

\bibitem{fw13}
Michal Feldman and Yoav Wilf.
\newblock Strategyproof facility location and the least squares objective.
\newblock In {\em Proceedings of the Fourteenth ACM Conference on Electronic
  Commerce}, EC '13, pages 873--890, New York, NY, USA, 2013. ACM.

\bibitem{moulin80}
Herve Moulin.
\newblock On strategy-proofness and single-peakedness.
\newblock {\em Public Choice}, 35:437--455, 1980.

\bibitem{pt13}
Ariel~D Procaccia and Moshe Tennenholtz.
\newblock Approximate mechanism design without money.
\newblock {\em ACM Transactions on Economics and Computation}, 1(4):18, 2013.

\end{thebibliography}

\end{document}